\documentclass[%
 reprint,
nofootinbib,
 amsmath,amssymb,
 aps,
 prd,
]{revtex4-1}

\usepackage{fancyhdr}
\usepackage{multirow}
\usepackage{amssymb}
\usepackage{amsmath,enumerate,mathtools}
\usepackage{amsthm}
\usepackage{wasysym}
\usepackage{enumitem}
\usepackage{tensor}
\usepackage{xcolor}
\usepackage{mathrsfs}
\usepackage[hidelinks]{hyperref}

\usepackage{graphicx}
\usepackage{dcolumn}
\usepackage{bm}

\newtheorem{theorem}{Theorem}[section]
\newtheorem{lemma}[theorem]{Lemma}

\newtheorem{corol}[theorem]{Corollary}

\newtheorem*{prop*}{}

\theoremstyle{definition}
\newtheorem{definition}[theorem]{Definition}

\newtheorem{remark}[theorem]{Remark}

\newcommand*\tho{\text{\thorn}}

\newcommand{\vphi}{\varphi} 
\newcommand{\nn}{{\nonumber}}
\newcommand{\kap}{{\kappa}}
\newcommand{\del}{{\delta}}

\DeclareMathOperator{\Tr}{Tr}

\newcommand*\DE{\ensuremath{\mathscr{D}}}

\newcommand*\de{\ensuremath{\textnormal{d}}}

\begin{document}

\title{Gauge fields with vanishing scalar invariants}


\author{Martin Kuchynka}
\email[]{kuchynkm@gmail.com}
\affiliation{Institute of Mathematics of the Czech Academy of Sciences \v Zitn\' a 25, 115 67 Prague 1, Czech Republic\\
 Institute of Theoretical Physics, Faculty of Mathematics and Physics, Charles University in Prague, \\V Hole\v{s}ovi\v{c}k\'ach 2, 180 00 Prague 8, Czech Republic}

\begin{abstract}
We consider extension of some established techniques of study of tensor fields on Lorentzian manifolds of arbitrary dimension to non-Abelian gauge covariant fields.  These are then applied to study of gauge fields with vanishing scalar curvature invariants and universal Yang-Mills fields, for which all possible higher-order corrections to the Yang-Mills equation identically vanish. 

\begin{description}
\item[DOI]
10.1088/1361-6382/ab4360
\end{description}
\end{abstract}

\maketitle


\section{Introduction\label{intro}}
During the 20th century, many sophisticated methods, techniques and formalisms were developed for study of tensor fields on four-dimensional Lorentzian manifolds. 
Some of them were later successfully extended to higher dimensions, amongst which being also algebraic classification of tensors based on null alignment \cite{alignment,review} and theory of balanced tensors \cite{bal1,VSIinHD,typeIIINuniversal,universalMaxwell}. Although very useful in countless applications (see e.g. the review \cite{review} and references therein), both mentioned techniques were thus far applied only to purely tensorial fields. Here, we consider their extension to more general non-Abelian gauge covariant fields. Their possible application in non-Abelian gauge theories is then demonstrated on two problems. 

The first one is a characterization of gauge fields $A_\mu$ with $VSI$ (vanishing scalar invariant) curvature (or field strength) $F_{\mu \nu}$, i.e. such that any scalar polynomial gauge invariant constructed from $F_{\mu \nu}$ and its gauge covariant derivatives of arbitrary order vanishes (cf. definition \ref{CSIVSIdef}). For brevity, we will sometimes refer to such $A_{\mu}$ as a $VSI$ gauge field throughout the paper.  $VSI$ and more general $CSI$\footnote{The corresponding scalar invariants do not necessarily vanish, but are constant.} tensor fields (mainly the Riemann curvature tensor and very recently also $p$-form fields) have already been extensively studied in the literature \cite{bal1,VSIinHD,HDVSI,CSIsoaces,hervikalignment,VSIelmag}. 
Both $VSI$ and $CSI$ fields proved to be useful in numerous applications to theories of gravity \cite{VSIsugra,CSIsugra,typeIIINuniversal,typeIIuniversal,CSIuniversal,Gorilla} and electrodynamics \cite{quantumelmag,universalMaxwell,EMuniversal}  with perhaps the most prominent role playing in the context of universal solutions to higher-order theories.   
It is thus natural to expect $VSI$ gauge fields to play a similar role in the context of higher-order non-Abelian gauge theories - in fact, some aspects of this role of $VSI$ gauge fields are addressed in the second problem studied in the paper. 

Concerning the characterization of $VSI$ gauge fields, we first obtain sufficient conditions for a $VSI$ gauge field corresponding to a general finite-dimensional gauge group (theorem \ref{FVSI}). Then, restricting ourselves to the case of a compact and semi-simple gauge group, we prove that these conditions are also necessary, arriving at the main result of section \ref{VSIcurvature} (theorem \ref{VSIchar}): 
\begin{prop*}[\textbf{Characterization of $VSI$ gauge fields}] Let the gauge group $G$ be compact and semi-simple. Then, a non-vanishing field strength $F_{\mu \nu}$ is $VSI$ if and only if 
$F_{\mu \nu}$ is null and aligned in a degenerate Kundt spacetime.
\end{prop*}
\noindent
This result is very reminiscent of characterization of $VSI$ spacetimes obtained in \cite{VSIinHD} and fully analogous to characterization of $VSI$ $p$-forms in \cite{VSIelmag}. A more detailed discussion can be found in section \ref{VSIcurvature}. Explicit form of gauge fields with aligned null curvature in degenerate Kundt spacetimes (defined e.g. in \cite{VSIelmag}) is then discussed in section \ref{explicitVSI}.

The rest of the paper is then devoted to the second problem, which is a study of \textit{universal} Yang-Mills (YM) fields in arbitrary dimension. 
By a universal YM field, we understand a field strength $F_{\mu \nu}$ of a gauge field $A_\mu$ such that all Lie algebra-valued 1-form polynomials at least quadratic in $F_{\mu \nu}$ or containing its gauge covariant derivatives vanish identically, cf. definition \ref{universaldef}.

Such fields are thus exact solutions not only to the field equation of the Yang-Mills theory, but also of any theory (described by Lagrangian \eqref{theory} at the beginning of section \ref{universalYMfields}) consisting of the Yang-Mills Lagrangian and appropriately defined gauge invariant higher-order corrections. Such corrections appear, most notably, in low-energy effective actions of string theories \cite{witten,tseytlinF1,tseytlinF2,coletti,YMstringcorrections}, where their explicit form may not be even known\footnote{In certain cases, the structure of $\alpha^\prime$ corrections is known at least to some extent. In the case of heterotic superstring theory, the structure of all these corrections is well understood (see e.g. \cite{SSW}, where this knowledge was successfully employed to prove $\alpha^\prime$-exactness of supersymmetric string waves).  On the other hand, in the case of an open string, it is only known that the part of the tree-level effective non-Abelian gauge field Lagrangian depending only on $F_{\mu \nu}$ and not on its gauge covariant derivatives is given by a symmetrized trace of the Born-Infeld Lagrangian \cite{BITseytlin}, while full form of the other parts of the effective Lagrangian containing also derivatives of $F_{\mu \nu}$ or going beyond tree level, remains unknown.}, but also in the context of various generalizations of YM theories (such as Lovelock-type extension of the YM theory, \cite{YMlovelock2} and references therein). 
It is thus appealing to have a method that will, within solutions of given low-energy effective field equations, identify the ones solving all such higher-order theories automatically irrespective of their particular form. 

In the case of Abelian gauge fields, perhaps one of the first results on universality was the observation made already in 30's and 40's  \cite{schrodinger1,schrodinger2} that any null $2$-form $F_{\mu \nu}$ in a flat spacetime solving source-free Maxwell equations in four dimensions is automatically a solution to any source-free non-linear electrodynamics with Lagrangian 
\begin{equation*}
\mathcal{L}(F_{\mu \nu} F^{\mu \nu}, F_{\mu \nu} \star F^{\mu \nu}). 
\end{equation*}
Further significant progress in studying universality of Maxwell fields and their $p$-form generalizations has been made rather recently \cite{quantumelmag,universalMaxwell}, incorporating the methods of null alignment and balanced tensors.
Employing an appropriate extension of these methods for gauge covariant fields, we succeed in obtaining sufficient (along with some necessary) conditions on universality for \textit{null} YM fields (defined e.g. in \cite{trautman1983,tafel} and corresponding to type N fields in classification provided in section \ref{prelim}) analogous to those of \cite{universalMaxwell}. The main result of section \ref{universalYMfields} can be then stated as (theorem\ref{universalthm})
\begin{prop*}[\textbf{Sufficient conditions for universality}] 
An aligned null Yang-Mills field in a degenerate Kundt spacetime of Weyl and traceless Ricci type III is universal and hence solves any theory \eqref{theory}. 
\end{prop*}
\noindent
The theorem thus extends Güven's earlier observation\footnote{Some early comments on universality of YM plane waves in flat space were given by Deser already in \cite{Deser}.} that YM plane waves in certain type N $pp$-wave backgrounds are immune to any polynomial higher-order corrections to YM equations, which was made already in \cite{guven} as a part of proving string $\alpha^\prime$-exactness of plane wave solutions to the ten-dimensional $\mathcal{N}=1$  supergravity. 
Since we do not deal with corrections other than the ones to the YM equations here, it enables us to extend the universality result to a general null YM field in a much broader class of background spacetimes including all $VSI$ spacetimes, all Weyl and traceless Ricci type III $pp$-waves, all Weyl and traceless Ricci type III $CSI$ Kundt spaces and even some non-CSI spacetimes. Explicit form of all these solutions in adapted coordinates is given in subsection \ref{examples}. 

The paper is organized as follows. In section \ref{prelim}, we first briefly review some tensorial notions such as boost weight, algebraic types or balanced tensors and related results. Then, gauge theoretical setting and notation is established. Lastly, extension of the corresponding tensorial notions and results to non-Abelian gauge covariant fields takes place. 
In section \ref{VSIcurvature}, $VSI$ property of gauge fields is studied and compared with already known characterization of $VSI$ metrics \cite{VSIinHD} and $VSI$ $p$-forms \cite{VSIelmag}.  
Section \ref{universalYMfields} is devoted to universal YM fields. Firstly, higher-order corrections to Yang-Mill Lagrangian under consideration are appropriately defined and reduction of the resulting field equations for $CSI$ gauge fields is discussed. 
The rest of the section then consists of study of null YM fields for which all previously defined higher-order corrections identically vanish. 
Explicit examples of universal solutions in adapted coordinates are also provided and followed by discussion of special cases already known in the literature. 
Although many of the results obtained in this section are a mere extension of some of the results already known for Maxwell fields to their non-Abelian counterparts, we believe their proofs provide a nice demonstration of the techniques presented in section \ref{prelim}. 
 In the appendix, further results on $k$-balanced gauge covariant fields can be found as well as GHP form of the YM equation and the Bianchi identity for $F_{\mu \nu}$ together with some of their consequences.

\section{Preliminaries}\label{prelim}

For the purposes of the rest of the paper, we will first recall some basics of algebraic classification of tensors based on null alignment developed in \cite{alignment} and reviewed in \cite{review}. Subsequently, theoretical setting and notation for gauge and gauge covariant fields will be established. 
The rest of the chapter is then devoted to extension of some crucial tensorial notions and related results to gauge covariant fields. 

On a domain of a $d$-dimensional spacetime, $d\geq 4$, with Lorentzian metric $g$ of signature $(-,+,\dots,+)$, we shall consider a \textit{null frame} $\{ e_{(0)} = \ell,  e_{(1)} = n, e_{(i)} = m_{(i)} \}$ consisting of a pair of null vector fields $\ell,n$ and $d-2$ spacelike vector fields $\{ m_{(i)} \}$ with $i = 2,\dots,d-1$. 
These are then subject to the Lorentz invariant orthogonality relations 
\begin{align}\label{ortogonality}
\ell_\mu  n^\mu = 1, \quad  m_\mu^{(i)}  m_{(j)}^\mu = \delta_j^i,
\end{align}
with the rest of the contractions vanishing.
Directional derivatives with respect to the individual null frame vectors are then denoted as $D \equiv \ell^\mu \nabla_\mu$, $\Delta \equiv n^\mu \nabla_\mu$, $\delta_i \equiv m_{(i)}^\mu \nabla_\mu$.

\subsection{Brief review of tensorial algebraic classification}\label{tensclass}
Let us briefly summarize a few basic notions of algebraic classification based on null alignment. 
Cf. also definitions 2.1, 2.2 and 2.6 of \cite{review}. 
 
\subsubsection*{Boost weight}
We say that a quantity $q$ has a boost weight $b$ (in the chosen null frame) if it transforms as $q \mapsto \lambda^b q$ under Lorentz boosts 
$\{ \ell, n, m_{(i)} \} \mapsto \{ \lambda \ell, \lambda^{-1} n, m_{(i)} \}$.

In particular, one can assign boost weight $b(T_{(a) \dots (b)})$ to null frame components $T_{ a \dots b } \equiv T_{\mu \dots \nu} e_{(a)}^\mu \dots e_{(b)}^\nu$ of any tensor $T$, being equal simply to a difference between the total number of 0's and 1's in the null frame indices $a, \dots, b$.   

\subsubsection*{Boost order}
To a tensor $T$, we assign a boost order $\textnormal{bo}(T)$ w.r.t. the null frame $\{ e_{(a)} \}$ as the highest boost weight $b(T_{a \dots b})$ over all its non-vanishing frame components $T_{ a \dots b }$.

In fact, it can be shown (Proposition III.2 of \cite{alignment}) that the values of $\textnormal{bo}(T)$ depends only on the choice of $e_{(0)}$ with the rest of the null frame being arbitrary.   
And, if there exists a vector $\ell \equiv e_{(0)}$ with respect to which $\textnormal{bo}(T) < b_{max}$, where $b_{max}$ is a maximal possible value of boost order of $T$ (depending on rank of $T$ and its symmetries), then $\ell$ is referred to as \textit{aligned null direction} (AND).

\subsubsection*{Boost order decomposition}
At this point, one can perform the boost order decomposition of any tensor $T$
\begin{equation}
T= \sum_b T_{(b)},
\end{equation}
defining the boost weight $b$ part $T_{(b)}$ of $T$ as the sum of all $T_{a \dots b} e^{(a)} \dots e^{(b)}$ with $b(T_{a \dots  b}) = b$. 
The index $b$ then ranges from $-b_{max}$ to $b_{max}$, where the value of $b_{max}$ depends on rank and symmetries of $T$. 

\subsubsection*{Algebraic types}
Given a non-vanishing tensor $T$, its algebraic type can be determined in the following manner:
\begin{itemize}
\item type G: $T_{(b_{max})} \neq 0$ in all null frames; 
\item type I: there is a null frame s.t. $T_{(b_{max})} = 0$; 
\item type II: there is a null frame s.t. $T_{(b)} = 0$ for all $b>0$; 
\item type D: there is a null frame s.t. $T_{(b)} = 0$ for all $b \neq 0$; 
\item type III: there is a null frame s.t. $T_{(b)} = 0$ for all $b\geq 0$; 
\item type N: there is a null frame s.t. $T_{(b)} = 0$ for all $b> -b_{max}$;   
\end{itemize}

In this paper, we will employ the presented algebraic types to classify spacetimes via algebraic types of its Weyl and Ricci tensors. 
In particular, we say that the spacetime is of Weyl type X and traceless Ricci type Y if its Weyl tensor is of type X and traceless part of its Ricci tensor is of type Y.

\subsection{Gauge fields: conventions and useful relations}
Let $G$ denote a real non-Abelian Lie group of a finite dimension and $\mathfrak{g}$ the associated Lie algebra spanned by generators $\{ t_{a} \}$, $a = 1, \dots, \dim \mathfrak{g}$.\footnote{Note the abuse of notation: we use Latin indices for labeling the generators of $G$ as well as for labeling the null frame vectors (for which the indices are running from 0 to  $d-1$).}
The generators of $G$ satisfy the commutation relations
\begin{equation}
[t_a, t_b] = i f\indices{^c_{ab}} t_c,
\end{equation}
where $ f\indices{^c_{ab}}$ are real structure constants of $\mathfrak{g}$.

\subsubsection*{Gauge field and field strength}
Introduction of a $\mathfrak{g}$-valued gauge field $A_{\mu} = A^a_{\mu} t_a$ gives rise to the gauge covariant derivative $\DE_\mu$ associated with $A_\mu$ 
\begin{equation}
\DE_\mu \equiv \nabla_\mu - i[A_\mu,\cdot],
\end{equation}
and the field strength (or curvature) $F_{\mu \nu} \equiv \nabla_{\mu} A_\nu - \nabla_\nu A_\mu -i [A_\mu,A_\nu]$.
The field strength of $A_\mu$ is then subject to the Bianchi identity
\begin{equation}\label{bianchi}
\DE_\mu F_{\nu \rho} + \DE_{\nu} F_{\rho \mu} + \DE_\rho F_{\mu \nu} = 0.
\end{equation}
Under a gauge transformation $U \in G$, the gauge field and its field strength undergo the following change
\begin{align}
A_\mu &\mapsto UA_\mu U^\dagger +i U \partial_\mu U^\dagger, \label{Atrans} \\ 
F_{\mu \nu} &\mapsto U F_{\mu \nu} U^\dagger.\label{Ftrans}
\end{align}

\subsubsection*{Gauge covariant fields}
We will say that a field $T_{\mu \dots \nu}$ is \textit{gauge covariant} if it transforms under any $U\in G$ as 
\begin{equation}\label{g-invariant}
T_{\mu \dots \nu} \mapsto U T_{\mu \dots \nu } U^\dagger.
\end{equation}
Gauge covariant derivative of a gauge covariant quantity $T_{\mu \dots \nu}$ is also gauge covariant, i.e. 
\begin{equation}\label{Dtransform}
\DE_\rho T_{\mu \dots \nu} \mapsto U \DE_\rho T_{\mu \dots \nu} U^\dagger.
\end{equation}
Another useful relation is the commutator relation for gauge covariant derivatives 
\begin{align}\label{gaugecommutator}
\begin{split}
[\DE_\sigma, \DE_\rho] T_{\mu_1 \dots \mu_k} =\ & \sum_i R\indices{^{\beta}_{\mu_i \rho \sigma}} T_{\mu_1 \dots \beta \dots \mu_k}\\ 
&- i[F_{\sigma \rho}, T_{\mu_1 \dots \mu_k}], 
\end{split}
\end{align} 
for any $\mathfrak{g}$-valued tensor $T_{\mu \dots \nu}$. 
If $F$ is a YM field, i.e. a solution to the source-free YM equation $\DE^\mu F_{\mu \nu} = 0$, the Bianchi identity \eqref{bianchi} together with \eqref{gaugecommutator} give us the following relation that will be useful later
\begin{align}\label{YMbox}
\begin{split}
\DE^\alpha \DE_\alpha F_{\mu \nu} =\ &2 R_{\alpha \mu \nu \beta} F^{\alpha \beta} + 2R_{\alpha [\nu}F\indices{_{\mu]}^\alpha}\\ 
&+ 2i [F_{\mu \alpha},F\indices{^\alpha_\nu}].
\end{split}
\end{align}

\subsection{Algebraic classification of gauge covariant fields}

There is a number of approaches to classification of YM fields or more general gauge fields on four-dimensional Lorentzian manifolds proposed in the literature, see e.g. \cite{gaugeclass} and references therein. 
These are based either on classification of polynomial Lorentz/gauge invariants or on the eigenvalue problem for $F_{\mu \nu}$. 
Our approach in this section is based on null alignment and hence is not restricted only to gauge fields and four dimensions. 
Of course, when applied to $F_{\mu \nu}$, there is an overlap of our classification with some of those already proposed for gauge fields in the literature. 
A detailed comparison is out of the scope of the paper, but some aspects of the overlap will be apparent at the end of this section and in section \ref{VSICSIfields}. 

Let $T_{\mu \dots \nu}=T_{\mu \dots \nu}^a t_a$ be a gauge covariant quantity. Given a null frame, one can assign boost order $\textnormal{bo}(T)$ to $T$ as 
\begin{equation}
\textnormal{bo}(T_{\mu \dots \nu}) \equiv \max_{a} \{ \textnormal{bo}(T_{\mu \dots \nu}^a) \},
\end{equation}
where $\textnormal{bo}(T^a)$ is boost order of a tensor $T^a$. 
If $T_{(b)}^a$ denotes the boost weight $b$ part of a tensor $T^a$, then $T_{(b)} \equiv T_{(b)}^a t_{a}$ corresponds to
the (gauge covariant) boost weight $b$ part $T_{(b)}$ of $T$. 

This enables one to classify any gauge covariant field $T$ exactly as in 
section \ref{tensclass}.
It is easy to see that such definition is equivalent to the following, more practical one. 
\begin{definition}
We will say that a gauge covariant tensor $T_{\mu \dots \nu}$ is of algebraic type X if all its Lie algebra components $T_{\mu \dots \nu}^a$ 
are aligned and of type X.   
\end{definition}

\begin{remark}
The notion of algebraic type of gauge covariant quantities is gauge invariant. Indeed, given a frame and $b$, 
$T_{(b)}$ vanishes if and only if $T_{(b)}^\prime$ of a gauge transformed field $T^\prime = U T U^\dagger$ vanishes as well. 
\end{remark}
Hence, algebraic type of a gauge covariant quantity is well-defined in this sense. 
Note that e.g. algebraic type of $A_\mu$ is not defined, since it transforms as \eqref{Atrans}, and hence, in general, the alignment type would not be preserved.

In the paper, we are focused mainly on \textit{null} field strengths, i.e. those of type N. Note that this definition of null field strength is consistent with the one of \cite{trautman1983,tafel}. 
In four dimensions, null solutions of the YM equation have already been studied extensively, see e.g. 
\cite{null1,tafel,null2} and references therein.

\subsection{Balanced gauge covariant fields}

One of the key ingredients in studying invariants, universality and proving various related results for tensor fields on Lorentzian manifolds is the notion of balanced tensors \cite{bal1,VSIinHD,typeIIINuniversal,universalMaxwell} which, in appropriate background 
spacetimes (the so called \textit{degenerate Kundt spacetimes}, see e.g. \cite{hervikKundt}), guarantees that all covariant derivatives of the field are well-behaved. 
Therefore, let us first provide appropriate extension of the notion of $k$-balancedness to gauge covariant fields. 
\begin{definition}
Let $k \in \mathbb{N}_0$. In a frame parallely propagated along a null geodetic affinely parameterized vector field $\ell$, 
a gauge covariant quantity $T$ is said to be $k$-balanced if only its boost weight $b<-k$ parts $T_{(b)}$ are possibly non-vanishing and satisfy $D_A^{-b-k} T_{(b)} = 0$, where $D_A \equiv \ell^\mu \DE_{\mu}$.  If $T$ is $0$-balanced, we will simply say it is balanced.
\end{definition}

\begin{remark}
From \eqref{g-invariant} and \eqref{Dtransform}, it can be seen that the notion of $k$-balancedness is again gauge invariant.
\end{remark}

In the appendix \ref{balancedappendix}, we have proven a number of results on balanced gauge covariant fields. 
Here, we present only the following three statements that will be of useful in the following sections. 
\begin{lemma}\label{balder}
Let $g$ be a degenerate Kundt spacetime. If a field strength $F_{\mu \nu}$ and its derivative $\DE_\rho F_{\mu \nu}$ are both aligned and of type II, 
then gauge covariant derivative of $k$-balanced gauge covariant quantity $T_{\mu \dots \nu}$ is again $k$-balanced. 
\end{lemma}
By induction, one has that gauge covariant derivatives of $T$ of arbitrary order are all aligned with $\ell$ and of boost order $b<k$. 
Note that, in the case of $F_{\mu \nu}$ being null, the Bianchi identity \eqref{bianchi} implies $D_A F_{\mu \nu} = 0$. Thus, for an aligned null $F_{\mu \nu}$ in a degenerate Kundt spacetime, conditions of \ref{balder} are automatically satisfied and we conclude
\begin{lemma}\label{balF}
Let $g$ be a degenerate Kundt spacetime. Then, aligned field strength $F$ is balanced if and only if it is null. 
\end{lemma}
In view of lemma \ref{balder}, one immediately arrives at   
\begin{corol}\label{balFder}
Let $F_{\mu \nu}$ be an aligned null field strength in a degenerate Kundt spacetime. Then, its gauge covariant derivatives of arbitrary order are aligned with $F_{\mu \nu}$ and of type III. 
\end{corol}
And since Hodge duality does not alter alignment, algebraic type nor behavior with respect to $D_A$, the conclusions hold also for the Hodge dual $\star F_{\mu \dots \nu}$ and its gauge covariant derivatives.

\subsection{Gauge covariant $VSI$ and $CSI$ fields}\label{VSICSIfields}
The notion of $VSI$ and $CSI$ spacetimes has been around for some time, see e.g. \cite{VSIinHD,CSIsoaces} and has proven useful in a number of applications. Lately, also $VSI$ and $CSI$ $p$-forms gained interest, mainly in applications in the context of universal Maxwell fields, cf. \cite{VSIelmag,universalMaxwell,quantumelmag}. 
Here, we consider a natural extension of the notion to gauge covariant fields.   
\begin{definition}\label{CSIVSIdef}
We will say that $T_{\mu \dots \nu}$ is $CSI_k$ if all scalar polynomial gauge invariants constructed\footnote{In the construction, the spacetime metric and volume element are allowed.} from $T_{\mu \dots \nu}$ and its gauge covariant derivatives up to order $k$ are constant. If $T_{\mu \dots \nu}$ is $CSI_k$ for all $k$, we will say it is $CSI$. Analogically, we will also define $VSI_k$ and $VSI$ fields as a subset of $CSI_k$ and $CSI$ fields, respectively, for which the corresponding invariants vanish. 
\end{definition}

Let us start with some simple lemmas on arbitrary gauge covariant fields, which will be particularly useful in obtaining characterization of $VSI$ field strengths. 

\begin{lemma}\label{VSIT}
If a gauge-covariant field $T_{\mu \dots \nu}$ is $VSI_k$, so is every tensor $U$ constructed as a trace of polynomial in $T_{\mu \dots \nu}$. In particular, $U$ and its covariant derivatives up to order $k$ are aligned and of type III. 
\end{lemma}
\begin{proof}
Let $H_{\mu \dots \nu}$ be any polynomial in $T_{\mu \dots \nu}$. 
Obviously, $U \equiv \Tr H$ has to be $VSI_0$, since any scalar polynomial constructed from $U$ is gauge invariant and therefore vanishes by assumption.  
It remains to notice that $\Tr [A_\rho,H_{\mu \dots \nu}]$ vanishes due to the cyclic property of $\Tr$, and hence
\begin{equation}
\nabla_\rho U_{\mu \dots \nu} = \Tr \DE_\rho H_{\mu \dots \nu}
\end{equation}  
is also gauge invariant. It is thus clear that also any scalar polynomial invariant constructed from U and its covariant derivatives of arbitrary order
is a gauge invariant and, in particular, $U$ has to be $VSI_k$. The algebraic VSI theorem \cite{hervikalignment} then guarantees that 
$\nabla^{(j)} U$ for all $j=0,1, \dots, k$ are aligned and of type III.
\end{proof}

From lemma \ref{VSIT} and lemma B.5 of \cite{VSIelmag}, it immediately follows
\begin{lemma}\label{Kundtlemma}
Let $T_{\mu \dots \nu}$ be a gauge covariant $VSI_2$ field. If there exists a non-vanishing rank-2 tensor constructed as a trace of polynomial in $T_{\mu \dots \nu}$, then $T_{\mu \dots \nu}$ is aligned with a Kundt vector.  
\end{lemma}

\section{Gauge fields with $VSI$ curvature}\label{VSIcurvature}
At this point, we are ready to prove sufficient conditions for gauge fields to have $VSI$ curvature and to obtain full characterization of such gauge fields in the case of compact and semi-simple gauge group. 

It is an immediate consequence of corollary \ref{balFder} that, in a degenerate Kundt spacetime, arbitrary polynomial constructed from balanced $F_{\mu \nu}$ and ist gauge covariant derivatives is of a negative boost order. Therefore
\begin{theorem}[Sufficient conditions for $VSI$ gauge fields]\label{FVSI}
If a field strength $F_{\mu \nu}$ is null, then it is $VSI_0$. 
If moreover $F_{\mu \nu}$ is aligned in a degenerate Kundt spacetime, then $F_{\mu \nu}$ is $VSI$.
\end{theorem}
\noindent 
Let us note that the first statement of theorem \ref{FVSI} follows in four dimensions also from \cite{gaugeclass}, where a basis for polynomial Lorentz invariants of $F_{\mu \nu}$ was given explicitly. 
It is also important to realize that, while the first part of the theorem holds actually for arbitrary gauge covariant $p$-form field, the second part of the theorem does not, since such $p$-form won't in general be balanced.

It turns out that, in the case of a compact semi-simple gauge group (for which the matrix $\kappa_{ab}\equiv \Tr t_a t_b$ is guaranteed to be non-degenerate and positive-definite), opposite direction of the statements of theorem \ref{FVSI} also holds. 
In particular, we have 
\begin{theorem}[Characterization of $VSI_0$ gauge fields]\label{VSI0char}
Let the gauge group $G$ be compact and semi-simple. Then, a field strength $F_{\mu \nu}$ is $VSI_0$ if and only if $F_{\mu \nu}$ is null. 
\end{theorem}
\begin{proof}
In view of theorem \ref{FVSI}, it remains to show that the $VSI_0$ property of $F_{\mu \nu}$ implies it is null. 
Let us thus consider $T_{\mu \nu} \equiv \Tr F_{\mu \rho}F\indices{_\nu^\rho}$. According to lemma \ref{VSIT}, $T_{\mu \nu}$ is of type III. 
Since all non-negative boost weight components $T_{00}, T_{0i}, T_{01}$ and $T_{ij}$ of $T_{\mu \nu}$ vanish and $\kappa_{ab}$ is non-degenerate and positive definite, the components $F_{0i}$, $F_{01}$ and $F_{ij}$ of $F$ have to vanish, and thus $F$ is null and aligned with $T_{\mu \nu}$. 
\end{proof}

Consequently, one can obtain also full characterization of gauge fields with $VSI$ curvature 
analogous to the one provided for $VSI$ spacetimes in \cite{VSIinHD} and $p$-forms  in \cite{VSIelmag}. 
\begin{theorem}[Characterization of $VSI$ gauge fields]\label{VSIchar}
Let the gauge group $G$ be compact and semi-simple. Then, a non-vanishing field strength $F_{\mu \nu}$ is $VSI$ if and only if 
$F_{\mu \nu}$ is null and aligned in a degenerate Kundt spacetime. 
\end{theorem}
\begin{proof}
At this moment, it is sufficient to prove the "only if" part of the statement and from theorem \ref{VSI0char}, we already have that $F_{\mu \nu}$ is null
and aligned with $T_{\mu \nu} \equiv \Tr F_{\mu \rho}F\indices{_\nu^\rho}$. 
 According to lemmas \ref{VSIT} and \ref{Kundtlemma}, $T_{\mu \nu}$ is $VSI$ and aligned with a Kundt vector, say $\ell$. 
It remains to show that the Kundt background spacetime $g$ is degenerate. It is thus sufficient to show that $R_{010i}$ and $D R_{0101}$ vanish \cite{hervikKundt}, cf. also Proposition A.2 of \cite{VSIelmag}.
For this purpose, let us consider the following tensors
\begin{align}
T_{\alpha \mu \beta \nu} &\equiv \Tr \DE_\alpha \DE_\gamma F_{\mu \rho} \DE_\beta \DE^\gamma F \indices{_\nu^\rho}, \\
T_{\alpha \beta \gamma \mu \delta \kappa \lambda \nu} &\equiv \Tr \DE_\alpha \DE_\beta \DE_\gamma F_{\mu \rho}
\DE_\delta \DE_\kappa \DE_\lambda F \indices{_\nu^\rho},
\end{align}
which have to be $VSI$ itself and, in particular, of type III.  
Employing a parallely propagated frame with affinely parameterized $\ell$, Ricci equations (A7)-(A12) and commutators (A13),(A14) of \cite{VSIelmag}, their boost weight zero components in the light cone gauge $A_\mu \ell^\mu = 0$ can be computed. 
For the first tensor, it is sufficient to consider 
\begin{align}
T_{0101} = DL_{1k} DL_{1k} \Tr F_{1j}F_{1j},
\end{align}
from which we conclude $D L_{1k} = 0$.  Boost weight zero component $T_{00110011}$ of the second tensor then reduces to
\begin{align}
T_{00110011} &= (D^2 L_{11})^2 \Tr F_{1i}F_{1i},
\end{align}
and thus also $D^2 L_{11} = 0$. Ricci equations (A7) and (A9) of \cite{VSIelmag} then imply that $R_{010i}$ and $DR_{0101}$ vanish, and hence $g$ is a degenerate Kundt spacetime \cite{hervikKundt}.   
\end{proof}

Comparing theorem \ref{VSIchar} with characterization of $VSI$ metrics, theorem 1 of \cite{VSIinHD}, we see that, in both cases, the curvature with vanishing scalar invariants is characterized by admitting a non-expanding and non-twisting geodetic aligned null vector $\ell$, along which it has to be of a negative boost order. Also note that, under the assumption of type III Riemann tensor, the Kundt spacetime is automatically degenerate. The situation in the case of gauge field curvature $F_{\mu \nu}$ is thus  analogous as in the case of the Riemann curvature $R_{\mu \nu \rho \sigma}$. 

When it comes to characterization of $VSI$ $p$-forms $F_{\mu \dots \nu}$, theorem 1.5 of \cite{VSIelmag} contains, apart from the common alignment assumptions, one additional condition: the Lie derivative of $F$ along $\ell$ has to satisfy $\pounds_\ell F = 0$. 
That's because, unlike the gauge field strength and the Riemann tensor, a general $p$-form $F$  does not have to satisfy the Bianchi identity and consequently $F$, although aligned and null in a degenerate Kundt spacetime, does not necessarily have to be balanced.  
In the case of $F$ satisfying $\de F=0$, the condition $\pounds_\ell F = 0$ is automatically satisfied from the null assumption (remark 1.8 of \cite{VSIelmag}) and the situation is analogous to the case of the gauge and the Riemann curvatures.

\subsection{Explicit form of $VSI$ gauge fields}\label{explicitVSI}
In this section, we discuss explicit form of gauge fields with null curvature in a degenerate Kundt spacetime. These are all necessarily $VSI$ (theorem \ref{FVSI}) and exhaust all $VSI$ gauge fields in the case of a compact and semi-simple gauge group $G$ (theorem \ref{VSIchar}).

Local form of a general $d$-dimensional degenerate Kundt metric in adapted Kundt coordinates $(r,u,x^\alpha)$ is given by 
\cite{hervikKundt}: 
\begin{align}\label{degKundt}
\begin{split}
\de s^2 = \: &2H(r,u,x) \de u^2 + 2W_\alpha(r,u,x)\de u \de x^\alpha \\
               &+2 \de u \de r   + g_{\alpha \beta}(u,x) \de x^\alpha \de x^\beta ,
\end{split}
\end{align}
 with $\alpha,\beta = 2, \dots , d-1$ and functions  $W_\alpha$ and $H$ at most linear and at most quadratic in $r$, respectively
\begin{equation}\label{Wform}
W_\alpha (r,u,x)= W^{(1)}_\alpha(u,x) r + W^{(0)}_\alpha(u,x),
\end{equation}
\begin{equation}\label{Hform}
H(r,u,x)=H^{(2)}(u,x) r^2 + H^{(1)}(u,x) r + H^{(0)}(u,x).
\end{equation}
Here, $r$ denotes affine parameter of the corresponding Kundt vector $\ell = \partial_r$. 

For any null field strength $F_{\mu \nu}$ aligned with $\ell$ in a degenerate Kundt spacetime \eqref{degKundt}, the only non-vanishing coordinate components are $F_{u \alpha}^a$:
\begin{equation}
F_{\mu \nu}^a \de x^\mu \de x^\nu = F_{u \alpha}^a (r,u,x) \de u \wedge \de x^\alpha.
\end{equation}
Exploiting gauge freedom, the form of the gauge potential $A_\mu$ can be then simplified considerably. 
Indeed, we can fix the light-cone gauge first to obtain $A_r^a=0$. Then, $F_{r \alpha}^a = 0$ reduces to $\partial_r A_\alpha^a = 0$ for all $\alpha$, which means that the remaining gauge freedom can be used to transform away also some of the $A_\alpha^a$. In fact, $F_{\alpha \beta} = 0$ guarantees that this can be done for all $A_\alpha^a$. Finally, $F_{ru}=0$ implies $\partial_r A_u = 0$ and we arrive at 
\begin{equation}\label{localgauge}
A_\mu^a \de x^\mu = A^a(u,x) \de u. 
\end{equation}
Any $VSI$ gauge field in the background spacetime \eqref{degKundt} can be then obtained by a gauge transformation of \eqref{localgauge}. 
Particular examples of $VSI$ gauge fields already given in the literature will be discussed in section \ref{examples} in the context of universal YM fields.

\section{Universal null Yang-Mills fields}\label{universalYMfields}

Let us first define \textit{universality} of Yang-Mills fields, which will be the subject of study in the rest of this section. 

\begin{definition}[Universal Yang-Mills fields]\label{universaldef}
Let $F_{\mu \nu}$ be a field strength of a gauge field $A_\mu$ and $k \in \mathbb{N}$. We say that $F_{\mu \nu}$ is $k$-universal if all  $\mathfrak{g}$-valued 1-forms constructed\footnote{In the construction of such polynomials, the spacetime metric and  volume element are allowed.} polynomially from $F$ (and its gauge covariant derivatives up to order $k$) that are at least quadratic in $F$ or contain $\DE F$, vanish identically. 
If $F_{\mu \nu}$ is $k$-universal for any $k \in \mathbb{N}$, we say it is universal. 
\end{definition}

A universal YM field $F$ thus solves the Yang-Mills equation as well as any of its generalizations involving higher-order corrections to YM equation constructed as polynomials of $F$, its Hodge dual and their gauge covariant derivatives of arbitrary order. Hence, the class of theories solved exactly by $F$ is rather broad and includes all Lagrangian theories considered in the following subsection.

\subsection{A class of generalized YM theories}
Let $G$ be a  semi-simple Lie group. Consider a theory with gauge group $G$ and Lagrangian $\mathcal{L}$ being a function of a finite set of gauge invariant scalars $\{ I_k \}$  
constructed as traces of polynomials in $F_{\mu \nu}$, its Hodge dual $\star F_{\mu \dots \nu}$ and their gauge covariant derivatives (one of the scalars being $I=\Tr F_{\mu \nu} F^{\mu \nu}$).  
Moreover, assume that $\mathcal{L}$ is analytic at zero and its power series expansion takes the form 
\begin{equation}\label{theory}
\mathcal{L} = \mathcal{L}_{YM} + \mathcal{L}_{HC},
\end{equation}
where $\mathcal{L}_{YM}$ is the Yang-Mills Lagrangian
\begin{equation}
\mathcal{L}_{YM} = -\frac{1}{2 \kappa} \Tr F_{\mu \nu} F^{\mu \nu},
\end{equation}
with $\kappa$ being a coupling constant 
and $\mathcal{L}_{HC}$ (\textit{higher-order corrections}) consists strictly of monomials 
at least cubic in $F_{\mu \nu}$ or containing its gauge covariant derivatives. 
Since $\mathcal{L}$ and $\mathcal{L}_{YM}$ are analytic function of scalar polynomials $I_k$, so is $\mathcal{L}_{HC}$.

\subsubsection{Reduction of EOM for $CSI$ fields}
It turns out that the form of equations of motion corresponding to \eqref{theory} reduce significantly for $CSI$ fields $F_{\mu \nu}$.  
This fact will be of great importance in the next section. 

Assume that the invariants $\{ I_m \}$ involved in the Lagrangian $\mathcal{L}$ are constructed from $F_{\mu \nu}$, $\star F_{\mu\dots  \nu}$ and their derivatives up to order $k$, i.e. schematically $I_m = \Tr \mathcal{P}_m (F, \DE F, \dots , \DE^k F)$ with $\mathcal{P}_m$ denoting the corresponding polynomial.  
Then, assuming the boundary terms vanish, variation of the action $S_{HC}$ corresponding to $\mathcal{L}_{HC}$ yields 
\begin{widetext}
\begin{align}
\delta S_{HC} &= \int \star 1  \sum_n \frac{\partial \mathcal{L}_{HC}}{\partial I_n} \delta I_n,  \\
&= \int \star 1 \sum_n \frac{\partial \mathcal{L}_{HC}}{\partial I_n} \Tr \bigg\{ \frac{\partial \mathcal{P}_n}{\partial F^{\mu \nu}} \delta F^{\mu \nu} +  \frac{\partial \mathcal{P}_n}{\partial \DE^\rho F^{\mu \nu}} \delta \DE^\rho F^{\mu \nu} + \dots  \bigg\},\\
&= \int \star 1  \Tr \bigg\{  \sum_n \frac{\partial \mathcal{L}_{HC}}{\partial I_n} \frac{\delta \mathcal{P}_n}{\delta A^{\nu}} \delta A^{\nu} \bigg\} + 
  \textnormal{terms involving } \nabla_\mu \frac{\partial \mathcal{L}_{HC}}{\partial I_n}. 
\end{align}
\end{widetext}
But terms involving covariant derivatives of $\partial \mathcal{L}_{HC} / \partial I_n$ will vanish for any $CSI_k$ field $F_{\mu \nu}$.  
Moreover, variation of any of the individual polynomials $\mathcal{P}_{n}$ takes the form $\DE^\mu {H}_{\mu \nu} + H_\nu$, where 
$H_{\mu \nu}$ and $H_\nu$ is some $2$-form and $1$-form, respectively, constructed polynomially from $F_{\mu \nu}$, $\star F_{\mu \dots \nu}$ and their gauge covariant derivatives (not necessarily only up to order $k$). Thus, the full field equations reduce for a $CSI_k$ field $F_{\mu \nu}$ (possessing invariants $\{ I_m \}$ within the domain of convergence of the Taylor series) to 
\begin{equation}\label{CSIEOM}
\DE^\mu  F_{\mu \nu} =  \DE^\mu  \tilde{F}_{\mu \nu} + F_\nu,
\end{equation} 
where $\tilde{F}_{\mu \nu}$ and $F_\nu$ are some polynomial gauge covariant $\mathfrak{g}$-valued forms with $F_\nu$ in general not expressible as a divergence of a $2$-form. 
Let us also note that the term $F_\nu$ has origin in variation of the non-Abelian contribution $[A,D^{m}F]$ to the gauge covariant derivative $D^{m+1} F$ and is thus present only in the non-Abelian case with $k>0$, i.e. when invariants $\{ I_n \}$ involve not only $F_{\mu \nu}$ (or $\star F_{\mu \dots \nu}$), but also its gauge covariant derivatives. 
In Abelian theories or non-Abelian theories involving only algebraic invariants $\{ I_n \}$, the field equations take the form \eqref{CSIEOM} with $F_\nu = 0$.

In the following sections, we will show that for certain types of $VSI_k$ solutions of the original Yang-Mill equations and certain background spacetimes, $F_{\mu \nu}$ is universal. In particular, above corrections necessarily vanish, and hence the field equations of any higher-order theory \eqref{theory} are automatically satisfied. Necessary conditions for universality of YM fields will be also studied.

\subsection{Universal solutions to Yang-Mills with algebraic corrections}\label{algebraiccorrections}

In the particular case of all invariants $\{ I_k \}$ being only algebraic (i.e. constructed solely from $F_{\mu \nu}$ and $\star F_{\mu \dots \nu}$ with no derivatives of these involved), $\mathcal{L}_{HC}$ represents \textit{algebraic corrections} to the YM Lagrangian and proposition 2.4 of \cite{universalMaxwell} can be extended to the non-Abelian case.

\begin{theorem}[Sufficient conditions for 0-universality]
Null Yang-Mills fields are $0$-universal and hence solve any theory \eqref{theory} with algebraic corrections $\mathcal{L}_{HC}$.
\end{theorem}
\begin{proof}
Since the YM field $F_{\mu \nu}$ is null, all polynomials at least quadratic in $F$ are of boost order at most $(-2)$ and hence any rank-1 contraction of such polynomial trivially vanishes, making $F$ $0$-universal. 

By theorem \ref{FVSI}, $F$ is $VSI_0$. The same holds also for $\star F_{\mu \dots \nu}$ and mixed invariants of both fields vanish as well. 
Consequently, as we saw in the previous section, the field equations of theory \eqref{theory} reduce to \eqref{CSIEOM} with $F_\nu = 0$ and 
with $\tilde{F}_{\mu \nu}$ constructed polynomially from $F_{\mu \nu}$ and its Hodge dual.   
But since  $\tilde{F}_{\mu \nu}$ has to be at least quadratic in $F_{\mu \nu}$ (recall that power series expansion of $\mathcal{L}_{HC}$ at zero consists of monomials at least cubic in $F_{\mu \nu}$), it vanishes due to $0$-universality of $F$.
\end{proof}

\subsection{Universal solutions to Yang-Mills with general corrections}\label{allcorrections}

In this section, we consider the general case of higher-order corrections constructed employing also gauge covariant derivatives of $F_{\mu \nu}$.   
Before proceeding further, let us first prove the following technical lemma on possible forms of skew-symmetric rank-2 contractions of gauge covariant derivatives $\DE_\mu \dots \DE_\nu F_{\alpha \beta}$. 
\begin{lemma}\label{contractions}
Let $F_{\mu \nu}$ be an aligned null YM field in a degenerate Kundt spacetime of Weyl and traceless Ricci type III. 
Then, any skew-symmetric rank-2 contraction of k-th gauge covariant derivative $\DE^k F_{\mu \nu}$ of $F_{\mu \nu}$ reduces to $F_{\mu \nu}$ 
multiplied by some polynomial in the Ricci scalar. 
\end{lemma}
\begin{proof}
Let $H_{\mu \nu}$ denote an arbitrary 2-form constructed by contractions (and antisymmetrization, if needed) of $\DE^k F_{\mu \nu}$.  
Clearly, to construct $H$, $k/2$ contractions of tensorial indices will be required and hence $k$ has to be even. 
Consequently, there are two possible types of terms in $H$ - either there is at least one contraction of some derivative index with a field strength index or all derivative indices are contracted.  

It turns out that, by commuting derivatives, these two types of terms in $\DE^k F_{\mu \nu}$ can be,  cast in the form 
\begin{equation}\label{forms}
\DE^{k-1}\DE^\mu F_{\mu \nu}, \qquad \DE^{k-2} \DE^\alpha \DE_\alpha F_{\mu \nu},
\end{equation}
respectively, for the price of producing new (but lower order) terms linear in $\DE^{k-2} F_{\mu \nu}$. 
Indeed, the formula \eqref{gaugecommutator} tells us that, schematically, $\DE^l [\DE,\DE] \DE^{m}F$ is given by $\DE^l [F,\DE^{m} F]$ and terms of type 
$\DE^l (\textnormal{Riem} \ast \DE^m F)$. However, for a balanced $F$, we have that $[F,\DE^{m}F]$ is of boost order $(-2)$ and hence won't contribute to $H$ and only boost weight (0) part of the Riemann tensor (given by the Ricci scalar $R$ and the metric) can survive in contribution from $\textnormal{Riem} \ast \DE^m F$. 
Thus, this term is simply proportional to $R \DE^m F$ and, since in a degenerate Kundt spacetime of Weyl and traceless Ricci type III, $\nabla_\mu R \propto\ell_\mu$ is of boost order $(-1)$ (see the proof of proposition A.8 in \cite{universalMaxwell}), $\DE^l (R \DE^m F)$ reduces to $R \DE^{k-2} F$. 
 
Now, each of these new terms proportional to $\DE^{k-2} F$ can be again, by the commutation procedure above, cast in the form \eqref{forms} introducing new lower-order terms proportional to $\DE^{k-4}F$. Therefore, the procedure above can be used iteratively until one ends up with all terms in $H$ 
of the form \eqref{forms} for some $k$ even. 
The first one obviously vanishes for any YM field, it thus remains to examine the second one. Employing the formula \eqref{YMbox}, 
we obtain 
\begin{equation}
\DE^\alpha \DE_\alpha F_{\mu \nu} \propto R F_{\mu \nu},
\end{equation}
where $R$ denotes the Ricci scalar. Recalling that $\nabla_\mu R$ is of boost order (-1), the second term in \eqref{forms} reduces to $ R \DE^{k-2} F_{\mu \nu}$ with all derivative indices contracted. 
As before, we can repeat the procedure above until we end up with all terms in the form $R^n F_{\mu \nu}$ for some $n$ up to a multiplication constant. Consequently, we obtain that $\tilde{F}$ can be actually always put in the form of $F$ multiplied by some polynomial in $R$. 
\end{proof}

\begin{remark}\label{symH}
From the proof of lemma \ref{contractions}, it is easy to see that, under assumption of the lemma, any \textit{symmetric} rank-2 contraction $H_{\mu \nu}$ of $\DE^k F$ can only consist of boost order (-2) terms of type $\DE^m F \otimes \DE^l F$, $m,l \geq 0$. Such terms would appear while commuting the gauge covariant derivatives in the procedure described above.
\end{remark}

\begin{theorem}[Sufficient conditions for universality]\label{universalthm}
An aligned null Yang-Mills field in a degenerate Kundt spacetime of Weyl and traceless Ricci type III is universal and hence solves any theory \eqref{theory}. 
\end{theorem}
\begin{proof}
Recall that, according to corollary \ref{balFder}, $F_{\mu \nu}$ and its gauge covariant derivatives of arbitrary order are aligned and of boost order at most (-1). Therefore, any nonvanishing 1-form constructed polynomially from these fields can be at most linear in them. 
Consequently, it is sufficient to prove that all rank-1 contractions $H_\mu$ of $\DE^k F$ vanish for $k>0$. 

There are two possible types of $H_\mu$ - either the index $\mu$ in $\DE_\mu \dots \DE_\nu F_{\rho \sigma}$ is left uncontracted, in which case $H_\mu = \DE_\mu I$, where $I$ is some $\mathfrak{g}$-valued tensorial invariant of $F$, or $\mu$ gets contracted, which means $H_\mu = \DE^\nu H_{\mu \nu}$ for some $H_{\mu \nu}$ constructed from $D^{k-1}F$. Obviously, as a boost order (0) object, $I$ has to vanish and consequently also the corresponding $H_\mu$. 
It remains to discuss the second possibility. We can decompose $H_{\mu \nu}$ into symmetric and skew-symmetric part and discuss their contributions to $H_\mu$ separately. 

According to remark \ref{symH}, the symmetric part of $H_{\mu \nu}$ consists of boost order (-2) terms of type $\DE^m F \otimes \DE^l F$ and hence cannot contribute to $H_\mu$, since such contribution can again only be of boost order  $(-2)$. 
On the other hand, lemma \ref{contractions} tells us that the skew-symmetric part of $H_{\mu \nu}$ takes the form $ \mathcal{P}(R) F_{\mu \nu}$ for some polynomial $\mathcal{P}$ of the Ricci scalar $R$. Thus, 
$H_\mu = F_{\mu \nu}  \nabla^\nu \mathcal{P}(R) + \mathcal{P}(R) \DE^\nu F_{\mu \nu}$. But $\nabla_\mu R \propto \ell_\mu$ in our background spacetime, and hence both terms in $H_\mu$ necessarily vanish for a null YM field. 

\end{proof}

\begin{theorem}[Necessary conditions for universality]
Let $F_{\mu \nu}$  be a non-vanishing null Yang-Mills field solving any theory \eqref{theory}. Then, $F_{\mu \nu}$ is either $CSI$ or it is aligned with gradient $\nabla_\mu I $ of all of its invariants $I$, these are all aligned and generate a twist-free geodetic null congruence. 
\end{theorem}
\begin{proof}
Consider a correction Lagrangian of the form $\mathcal{L}_{HC} \equiv I \mathcal{L}_{YM}$, where $I$ is arbitrary scalar qauge invariant of $F_{\mu \nu}$. Since the invariant $\mathcal{L}_{YM}$ vanishes for a null $F_{\mu \nu}$, varying $\mathcal{L}_{HC}$ and employing the YM equation, the requirement of universality reduces to
\begin{equation}\label{univpodm}
 F_{\mu \nu} \nabla^\mu I = 0, 
\end{equation}
for all invariants I. Since $F_{\mu \nu}$ is null and non-vanishing, decomposing the covariant derivative into null frame, $\nabla_\mu = \ell_\mu \Delta + n_\mu D + m_\mu^{(i)} \delta_i$, \eqref{univpodm} implies that both $DI$ and $\delta_i I$ vanish, hence $\nabla_\mu I = \ell_\mu \Delta I $. 
Therefore, gradients of all invariants I are aligned. Moreover, since the null vector $\ell$ is proportional to gradient of a function, it generates a twistfree geodetic congruence. 
\end{proof}

\begin{remark}
Note that, while in four dimensions, any null YM field is necessarily aligned with a shear-free geodetic null 
congruence \cite{tafel}, in higher dimensions, this is not the case even for Maxwell fields, cf. appendix B of \cite{ortaggioRT} and lemmas 3, 4 in \cite{ghpclanek}. 
In fact, the congruence has necessarily non-vanishing shear in higher-dimensions, unless it is non-expanding. 
Employing the YM equation and Bianchi identity in the generalized GHP formalism given by \eqref{max:1}-\eqref{max:4} in the appendix, 
we provide an appropriate extension of lemma 4 in \cite{ghpclanek} to non-vanishing null YM fields, see lemma \ref{MariotYM}.
 
However, under additional conditions on curvature of the background spacetime, the admissible geometries of the corresponding null congruence further reduce due to Ricci and Bianchi equations. In particular, in a spacetime of Weyl and traceless Ricci type N with constant Ricci scalar, any null YM field is necessarily aligned with shear-free, twist-free and non-expanding null geodetic (i.e. Kundt) congruence, as we prove in the appendix (theorem \ref{NNadmitsnullF}).  
\end{remark}

\subsection{Examples of universal Yang-Mills fields}\label{examples}

We first give explicit form of all universal solutions satisfying the hypothesis of theorem \ref{universalthm} in adapted coordinates. 
Then, discussion of special cases already known in the literature follows.  

\subsubsection*{The metric}
Recall that local form of a general degenerate Kundt metric in Kundt coordinates $(r,u,x^\alpha)$ is given by 
\eqref{degKundt}. 
Requirement that $g$ is of Weyl and traceless Ricci type III then puts further restrictions on functions in \eqref{degKundt}, cf. discussion in section 4 of \cite{EMtheory}.  
In particular, $g^T \equiv g_{\alpha \beta}(u,x) \de x^\alpha \de x^\beta$ has to be a Riemannian metric on a  $(d-2)$-dimensional  transverse space spanned by coordinates $\{ x^\alpha \}$, which is, for any fixed $u$, of constant sectional curvature $K(u)$ depending on the Ricci scalar $R=R(u)$ of \eqref{degKundt} and the dimension $d$ as
\begin{equation}\label{sectcurv}
K \equiv \frac{R}{d(d-1)},
\end{equation}
and functions $H$ and $W_{\alpha}$ are constrained by the following equations, where $\nabla^T$ denotes the covariant derivative associated with the transverse space metric $g^T$
\begin{equation}\label{constraint1}
2H^{(2)} = \frac{1}{4}W_\alpha^{(1)}W^{(1)\alpha} + K,
\end{equation}
\begin{equation}\label{constraint2}
\nabla_{(\beta}^T W_{\alpha) }^{(1)} - \frac{1}{2} W_\alpha^{(1)} W_\beta^{(1)} = 2K g_{\alpha \beta}, 
\quad \partial_{[\beta } W_{\alpha ] }^{(1)}=0.
\end{equation}

\subsubsection*{The gauge field}
Exploiting gauge freedom, we can without loss of generality start with gauge potential of the form \eqref{localgauge}. 
The function $A^a$ is then subject to the YM equation, which in our case reduces to the wave equation in the $(d-2)$-dimensional transverse space of constant sectional curvature \eqref{sectcurv}:
\begin{equation}\label{YMtransverse}
\Box^T A^a = 0,
\end{equation}
where $\Box^T \equiv g^{T \alpha \beta} \nabla_\alpha^T  \nabla_\beta^T$ is the Laplace-Beltrami operator. Any null YM field in the background spacetime \eqref{degKundt} can be then obtained by a gauge transformation of \eqref{localgauge} and is necessarily $VSI$ (recall theorem \ref{FVSI}). 

Note that the YM equation \eqref{YMtransverse} depends only on the transverse metric $g^T$, while the functions $H$ and $W_\alpha$ in \eqref{degKundt} can be arbitrary, as long as the constraints \eqref{constraint1}, \eqref{constraint2} hold.

\subsubsection{Examples in the literature}
Some of the solutions covered by  \eqref{degKundt} and \eqref{localgauge} are already well-known, most of them being plane waves in four-dimensional $pp$-wave background such as Coleman's non-Abelian plane waves in Minkowski spacetime \cite{coleman} or their generalization to curved $pp$-wave background \cite{guvenEYM}. Both of these examples correspond to the case $W_\alpha = 0$, $\partial_r H=0$ and $g^T$ being a flat metric (i.e. $K=0$).   

However, there are also some more general examples, namely Weyl type III Einstein-Yang-Mills solutions found in \cite{fusterEYM} and containing the mentioned plane wave solutions as a special case. These examples admit $W_\alpha  \neq 0$ and $H$ quadratic in $r$, while the transverse metric is still flat. In the Weyl type N limit, these solutions reduce to either $pp$-waves or Kundt waves, as discussed in the reference above. 

Some higher-dimensional solutions are also known, most notably the Weyl type N plane waves in ten-dimensional $\mathcal{N}=1$ supergravity considered by  Güven in \cite{guven}. These again correspond to the above fields with $W_\alpha = 0$, flat transverse metric $g^T$ and $r$-independent functions $H$ and $A^a$ in the plane wave ansatz. 
Some Non-Abelian waves propagating in higher-dimensional (A)dS wave backgrounds (with non-flat transverse metric $g^T$) are also known, cf. \cite{4Dsugra} and \cite{adswaveYM}.

As we see, \eqref{degKundt} and \eqref{localgauge} contain many already known solutions and their higher-dimensional generalizations (with all of the mentioned ones being $VSI$ gauge fields propagating in a $CSI$ background spacetime), but apparently also a new ones. 
Unlike the Maxwell case (see e.g. \cite{VSIelmag}), to our knowledge, non-Abelian YM fields in a degenerate Kundt spacetime with non-flat transverse metric $g^T$ were considered so far only in the special case of the (A)dS wave background such as the ones mentioned above.

\section*{Conclusions}
The purpose of the paper was twofold. Firstly, we have extended some of the established tensorial techniques to gauge covariant fields with potential application in non-Abelian gauge theories. 
Secondly, these techniques were applied to study of gauge fields with $VSI$ curvature and universal (test) Yang-Mills fields in arbitrary dimension.

Taking into account also backreaction of the YM field on spacetime geometry (and thus dealing also with $\mathcal{L}_{HC}$-corrections to EOM for the spacetime metric) would deserve further investigation. 
However, it is already clear at this point that, employing techniques presented in the paper, this can be done in the spirit of \cite{EMuniversal}, allowing also for a presence of appropriate scalar fields and $p$-forms (such as dilaton and KR or RR fields, respectively) along the lines of  \cite{guven} and thus being of particular interest in low-energy limits of string theories. 

Hopefully, the paper will serve as a clear explanation of the presented techniques as well as a guide for their use in applications  
and will perhaps lead to further new results on non-Abelian fields or even to development/extensions of other useful techniques in the future.

\begin{acknowledgments}
The author is thankful to M. Ortaggio for useful suggestions and comments on the draft. 
This work has been supported by Research Plan RVO: 67985840 and Research grant GAČR 19-09659S.
\end{acknowledgments}

\appendix
\section{On k-balanced gauge covariant fields}\label{balancedappendix}

Let $g$ be a degenerate Kundt metric. 
In the rest of the section, we will work in a paralelly propagated frame $\{ e_{(a)} \} \equiv \{ \ell,n,m_{(i)} \}$ (i.e. such that $D n_\nu = 0$ and $D m_\nu^{(i)} = 0$) with affinely parameterized Kundt vector $\ell$. 
 
In a parallely propagated frame with affinely parameterized $\ell$, $k$-balanced tensors are defined as boost order $b<-k$ tensors with null frame components behaving appropriately under the action of the Newman-Penrose directional derivative $D$. 
$k$-balanced gauge covariant fields are defined analogously, but with $D$ replaced by directional gauge covariant derivative $D_A \equiv \ell^\mu \DE_\mu$ associated with the gauge field $A_\mu$.    
However, since the notion of $k$-balancedness is gauge independent, we can easily overcome this obstacle by fixing the light-cone gauge 
\begin{equation}
\ell^\mu A_\mu =0,
\end{equation}
(recall that in the standard Kundt coordinates, $\ell = \partial_r$) in which case $D_A$ reduces to $D$. Consequently, balancedness for gauge covariant quantities can be restated in terms of balanced tensors \cite{universalMaxwell}. 

\begin{lemma}
A gauge covariant quantity $T_{\mu \dots \nu}$  in the light-cone gauge is $k$-balanced if and only if $\{ T_{\mu \dots \nu}^a \}$ are aligned $k$-balanced tensors. 
\end{lemma}

This simple fact enables one to employ tensorial results in proving that $\DE_{\mu}$ preserves balancedness of gauge covariant fields in degenerate Kundt spacetimes. 

If $T_{\mu \dots \nu}$ is $k$-balanced, then frame components of $\DE_\rho T_{\mu \dots \nu}$ in the light-cone gauge contain, in addition to frame components of $\nabla_\rho T_{\mu \dots \nu}^c$, only two more scalars 
\begin{equation}\label{balancedgauge}
 A_{1}^a f_{ab}^c \eta^b, \qquad  A_{i}^a f_{ab}^c \eta^b,
\end{equation}
 where $\eta^b$ are the frame components of $T_{\mu \dots \nu}^c$ and $A_{a} \equiv e_{(a)}^\mu A_\mu$. However, if $D A_{i} $ and $D^2 A_1 $ vanish, \eqref{balancedgauge} are again $k$-balanced scalars, i.e. each of these two is of boost weight $b^\prime<-k$ and vanishes under the action of $D^{-b^\prime-k}$. 
In the degenerate Kundt spacetime, the two conditions on $A_i$ and $A_1$ are in the light-cone gauge equivalent to $F_{0i}=0$ and $D F_{01} = 0$ (or in a gauge covariant way, $F_{\mu \nu}$ is of type II and $D_A F_{0 1}=0$) in a parallely propagated frame with affinely parametrized $\ell$.

In fact, it can be easily computed that for type II curvature $F_{\mu \nu}$ aligned with an affinely parameterized Kundt vector, $\DE_{\rho} F_{\mu \nu}$ is again of aligned type II if and only if $D_A F_{01}$ and $D_A F_{ij}$ both vanish in a parallely propagated frame. However, the Bianchi identity \eqref{bianchi} implies that $D_A F_{ij} = 0$ follows automatically from $F_{0i} = 0$ and hence
\begin{lemma}\label{typeIIDF}
Let $F_{\mu \nu}$ be a type II field strength aligned with an affinely parameterized Kundt vector. Then, $\DE_{\rho}F_{\mu \nu}$ is also of aligned type II if and only if $D_A F_{01}=0$ in a parallely propagated frame. 
\end{lemma}

In turn, conditions for balancedness of \eqref{balancedgauge} can be restated covariantly via alignment and algebraic type of $F_{\mu \nu}$ and $\DE_\rho F_{\mu \nu}$, which implies lemma \ref{balder}. 
Obviously, both conditions $F_{0i}=0$ and $D_A F_{01}=0$ are trivially satisfied if $F_{\mu \nu}$ is aligned and null.

\begin{corol}
Let $g$ be a degenerate Kundt spacetime and the field strength $F_{\mu \nu}$ be aligned and null. 
Then, gauge covariant derivative $\DE_{\rho} T_{\mu \dots \nu}$ of a $k$-balanced gauge covariant quantity $T_{\mu \dots \nu}$ is again $k$-balanced. 
\end{corol}

Let us stress again that the results, although proven in the light-cone gauge, are really gauge independent. 
It is also worth noticing that Bianchi identity \eqref{bianchi} implies $D_A F_{\mu \nu} = 0$ for any null field strength $F$ aligned with a Kundt vector $\ell$. Therefore,  we arrive at lemma \ref{balF} and corollary \ref{balFder}.

\section{Yang-Mills fields in the GHP formalism}
In \cite{newmanYM}, a formulation of Yang-Mills theory in four-dimensional Newman-Penrose formalism was given. 
Here, we provide a formulation of the theory in a higher-dimensional generalization of the GHP formalism developed in \cite{ghpclanek}. 
In the paper, the GHP form of the Maxwell equations for $p$-form fields is already presented, see equations (3.3) - (3.7) of \cite{ghpclanek}.   
Adopting analogous notation for the null frame projections of the gauge field $A_\mu = A_\mu^a t_a$ and its field strength $F_{\mu \nu} = F_{\mu \nu}^a t_a$, we define for all $a$ the following GHP scalars
\begin{equation}
A^a \equiv A_{0}^a, \qquad A_i^a \equiv A_{i}^a, \qquad A^{\prime a} \equiv A_{1}^a,
\end{equation}
\begin{equation}
\vphi_i^a \equiv F_{0i}^a, \quad f^a \equiv F_{01}^a, \quad F_{ij}^a \equiv F_{ij}^a, \quad \vphi_i^{\prime a} \equiv F_{1i}^a.
\end{equation}

The field strength GHP scalars can be expressed in terms of the ones associated with the gauge potential as
\begin{align}
\vphi_i^a &= \tho A_i^a - \eth_i A^a + \tau_i^\prime A^a + \kappa_i A^{\prime a} + \rho_{ji}A_j^a + A^b A_i^c f_{bc}^a,\label{ghpF1} \\
F_{ij}^a &=  \eth_i A_j^a - \eth_j A_i^a - 2A^a \rho_{[ij]}^\prime - 2A^{\prime a}\rho_{[ij]} + A_i^b A_j^c f_{bc}^a,\label{ghpF3}\\
f^a &= \tho A^{\prime a} - \tho^\prime A^a + (\tau_i - \tau_i^\prime)A_i^a  + A^{\prime b} A^c f_{bc}^a \label{ghpF2},
\end{align}
and  $\vphi_i^{\prime a}$ can be obtained from $\vphi_i^a$ via the priming operation.  
Various non-trivial null frame projections of the YM equation $\DE^\mu F_{\mu \nu} = 0$ and the Bianchi identity $\DE_{[\mu} F_{\nu \rho]} = 0$ then yield the following set of GHP equations

\begin{widetext}
\subsubsection*{Boost weight +1}
\begin{eqnarray}
   \eth_i \vphi_{i}^a + \tho f^a
        &=&  A_i^b \vphi_i^c f_{cb}^a +  A^b f^c f_{cb}^a +  \tau'_i \vphi_{i}^a - \rho f^a
            + \rho_{[ij]} F_{ij}^a - \kap_i \vphi_{i }^{\prime a},  \label{max:1}\\
  2 \eth_{[i} \vphi_{j]}^a - \tho F_{ij}^a
        &=&  A_{[i}^b \vphi_{j]}^c f_{cb}^a +  A^b F_{ij}^c f_{bc}^a + 2 ( \tau'_{[i} \vphi_{j]}^a + \rho_{k[i} F_{|k |j ]}^a
       +  \rho_{[ij]} f^a  + \kap_{[i} \vphi_{j]}^{\prime a} ),
        \label{max:2}
\end{eqnarray}
\end{widetext}
\begin{widetext}
\subsubsection*{Boost weight 0}
\begin{eqnarray}
  2\tho' \vphi_{i}^a + \eth_j F_{j i}^a -
        \eth_{i}f^a   &=& 2 A^{\prime b} \vphi_{i}^c f_{cb}^a +  A_j^b F_{j i}^c f_{cb}^a + A_i^b f^c f_{bc}^a 
+ 2\tau_j F_{j i}^a  - 2 \tau_{i} f^a   \nn \\ 
          & & +  ( \rho'_{i j} - \rho'_{j i} - \rho' \del_{ij})\vphi_{j}^a
            + ( \rho_{i j} + \rho_{j i} -\rho\del_{i j})\vphi_{j}^{\prime a},\label{max:3}\\
  \eth_{[i}F_{j k ]}^a &=& A_{[i}^b F_{jk]}^c f_{cb}^a +  2 (\vphi_{[i} \rho_{j k ]}^{\prime a}
                                       + \vphi_{[i}^{\prime a} \rho_{j k]} ), \label{max:4}
\end{eqnarray}
\end{widetext}
\noindent
together with the primed equations \eqref{max:1}$'$, \eqref{max:2}$'$ and \eqref{max:3}$'$. Beware of the fact that, under priming operation, the scalars $f^a$ transform as $f^{\prime a} = -f^a$.

Obviously, the GHP equations \eqref{max:1}-\eqref{max:4} differ from the GHP Maxwell equations in \cite{ghpclanek} only by presence of the additional terms emanating from the commutator $[A_\mu,F_{\nu \rho}]$ in $\DE_{\mu} F_{\nu \rho}$. Therefore, it  can be expected 
that some of the results obtained for Maxwell fields from the GHP Maxwell equations can be appropriately extended also to YM fields. 

Indeed, since for a null field strength $F_{\mu \nu}$, null frame $\{ \ell, n, m_{(i)} \}$ can be chosen such that only scalars $\vphi_i^{\prime a}$ are non-vanishing, equations \eqref{max:1} - \eqref{max:4} reduce to 
\begin{align}
 \kap_i \vphi_{i }^{\prime a} &= 0, &\kap_{[i} \vphi_{j]}^{\prime a} = 0,\\
( \rho_{i j} + \rho_{j i} -\rho\del_{i j})\vphi_{j}^{\prime a} &=0,   &\vphi_{[i}^{\prime a} \rho_{j k]} = 0.
\end{align}
Therefore, lemma 4 of \cite{ghpclanek} can be extended to YM fields in the following way 
\begin{lemma}\label{MariotYM}
If a non-vanishing null Yang-Mills field $F_{\mu \nu}$ is aligned with a null vector $\ell$, then $\ell$ is geodetic. Moreover, for its optical matrix $\rho_{ij}$ and the corresponding GHP scalars $\vphi_i^{\prime a}$ associated with $F_{\mu \nu}$, there exists $\omega_i^a$ such that
\begin{equation}\label{optical}
\rho_{(ij)}\vphi_j^{\prime a} = (\rho/2)\vphi_i^{\prime a}, \qquad 
\rho_{[ij]} = \omega_{[i}^a \vphi_{j]}^{\prime a}.
\end{equation}
\end{lemma}

\noindent
Note that, in four dimensions, the lemma reduces to the well-known result that any non-vanishing null YM field is aligned with 
shear-free geodetic congruence \cite{newmanYM,tafel}. 

In a similar way, more results on YM fields (both the test fields in a fixed background and the ones coupled to gravity) can be obtained. Here, we provide an elementary proof of the following result stated originally for Einstein-Maxwell solutions in arbitrary dimension \cite{EMtheory}.

\begin{theorem}\label{NNadmitsnullF}
If a spacetime of Weyl and traceless Ricci type N with constant Ricci scalar admits a non-vanishing null Yang-Mills field, then the spacetime is Kundt.
\end{theorem}
\begin{proof}
Due to proposition 3.1 of \cite{NNspacetimes}, the Weyl and the Ricci tensors (and consequently also $F_{\mu \nu}$) are aligned with the same null vector, say $\ell$. 
Moreover, according to proposition 5.4 of \cite{NNspacetimes}, null frame be chosen such that the optical matrix $\rho_{ij}$ takes the form such that only the following components are non-vanishing 
\begin{equation}
\begin{gathered}
\rho_{22}=s, \qquad \rho_{33}= s b,\\ 
\rho_{23}= s a, \qquad \rho_{32}= -s a,
\end{gathered}
\end{equation}
with $b\neq 1$, otherwise $F_{\mu \nu}$ would have to be zero. 
Thus, the expansion of $\ell$ reads $\rho = s (1+b)$ and the first equation of \eqref{optical} tells us that of the eigenvectors is $\rho/2$. This 
means that either $\rho = 0$ or $b=1$ and, since $b = 1$ is forbidden, we have $\rho =0$. Consequently, either $s=0$ (in which case the spacetime is Kundt and we are done) or $b=-1$.  Also, if $a = 0$, the spacetime is again Kundt due to the Sachs equation \cite{RicciHD}. 
Assume thus that $s,a \neq 0$ and $b=-1$. Then, the first equation of \eqref{optical} implies $\vphi_2^{\prime a},\vphi_3^{\prime a} = 0$. 
The second of \eqref{optical} then gives us $\vphi_i^{\prime a} = 0$ for all $i>3$. Hence $F_{\mu \nu}$ vanishes completely, which is a contradiction and the spacetime is indeed Kundt.
\end{proof}

Consider now a solution  $(g_{\mu \nu},F_{\mu \nu})$  of the Einstein-YM equations with a non-vanishing null field strength $F_{\mu \nu}$. 
Then, $g_{\mu \nu}$ is of traceless Ricci type N. Therefore, in view of theorem \ref{NNadmitsnullF}, we immediately obtain
\begin{corol}
All Weyl type N solutions of the Einstein-Yang-Mills equations with a non-vanishing null field strength are aligned and Kundt.
\end{corol}

Thus, although null YM test fields are not necessarily aligned with a shear-free null congruence in higher dimensions, it is still geodetic and, once 
additional assumptions on structure of curvature tensors of the background spacetime are taken into account, further restrictions on geometry of the corresponding null congruence follow.




%

\end{document}